\documentclass{ws-procs975x65}

\begin{document}

\title{Quasilocal mass and surface Hamiltonian in spacetime}

\author{Mu-Tao Wang}

\address{Department of Mathematics, Columbia University,\\
New York, NY 10027, USA\\
$^*$E-mail: mtwang@math.columbia.edu}

\begin{abstract}
We discuss the concepts of energy and mass in relativity. On a finitely extended spatial region, they lead to the notion of quasilocal energy/mass for the boundary 2-surface in spacetime. A new definition was found in \cite{wy1} that satisfies the positivity, rigidity, and asymptotics properties. The definition makes use of the surface Hamiltonian term which arises from Hamilton-Jacobi analysis of the gravitation action. The reference surface Hamiltonian is associated with an isometric embedding of the 2-surface into the Minkowski space. We discuss this new definition of mass as well as the reference surface Hamiltonian. Most of the discussion is based on joint work with PoNing Chen and Shing-Tung Yau.
\end{abstract}

\keywords{Quasilocal mass, surface Hamiltonian.}

\bodymatter

\section{Energy of matter fields and conservation}

Relativity is a unified theory of space and time. The spacetime of special relativity is the Minkowski space $\mathbb{R}^{3,1}=\mathbb{R}\times \mathbb{R}^3$ with Lorentz metric of signature $(-1, 1, 1, 1)$. We normalize  the speed of light to be $1$. The light cone consists of four vectors $(t, x, y, z)$ with $=-t^2+x^2+y^2+z^2=0$. As nothing travels faster than light, a material particle or an observer moves in future timelike direction.

To each matter field, an energy-momentum tensor $T$ is attached. $T$ is derived from the Lagrangian of the field and is described by first derivatives of the field. In particular, it is a symmetric $(0,2)$ tensor $T_{\mu\nu}$ which satisfies the conservation law

\begin{equation}\label{conserv}\nabla^\mu T_{\mu\nu}=0.\end{equation}

Without gravitation, the energy of a physical system $\Omega$ is obtained by integrating $T$ on $\Omega$ with respect to an observer. To be more precise, given a spacelike bounded region $\Omega$, the energy intercepted by $\Omega$ as seen by the observer $t^\nu$ is the flux integral
\[\int_\Omega T_{\mu\nu} t^\mu u^\nu \] where $u^\nu$ is the future timelike unit normal of $\Omega$. The dominant energy condition guarantees $T_{\mu\nu} t^\mu u^\nu\geq 0$.

Suppose $t^\mu$ is a constant future-directed timelike unit vector in $\mathbb{R}^{3,1}$. By conservation law \eqref{conserv}, $T_{\mu\nu}t^\mu$ is divergence free and thus is dual to a closed 3-form in $\mathbb{R}^{3,1}$, which in turn is $d\omega$ for a 2-form $\omega$.
 Therefore, $\int_{\Omega} T_{\mu\nu}t^\mu u^\nu=\int_{\partial \Omega}\omega $ is a linear expression in $t^\mu$.
 Minimizing among all such observers $t^\mu$ gives the quasilocal mass which depends only on the boundary 2-surface $\Sigma=\partial\Omega$. Moreover, $\int_{\Omega} T_{\mu\nu}u^\nu$ defines a quasilocal energy-momentum 4-vector. This is the prototype of quasilocal mass and quasilocal energy-momentum.

\section{Energy in General relativity}

In general relativity, spacetime is a 4-dimensional manifold with a Lorentz metric $g$, the gravitational field. Local causal structure of spacetime remains the same, and each tangent space is isometric to the Minkowski space. Gravitational force is represented by the spacetime curvature of $g$. The relation between the gravitation field and matter fields is exactly described by the Einstein equation

    \begin{equation}\label{einstein}Ric-\frac{1}{2} R g=8\pi T\end{equation}
     where $Ric$ is the Ricci curvature, and $R$ is the scalar curvature of $g$, respectively. $T$ represents the energy-momentum tensor of all matter fields. This is the Euler-Lagrange equation of the Hilbert-Einstein action.

Concerning energy, one seeks for an energy momentum tensor for gravitation. However, it turns out first derivatives of $g$ are all coordinate dependent, and thus there is no density for gravitational energy. This is Einstein's equivalence principle.
One can still integrate $T$ on the right hand side of \eqref{einstein} but this gives only the energy contribution from matters. Indeed, there exists vacuum spacetime, i.e. $T=0$, with nonzero energy such as Schwarzschild's or Kerr's solution of Einstein's equation. This is gravitational energy by the sheer presence of spacetime curvature. Even without energy density, one can still ask the question: what is the energy in a system $\Omega$, counting contribution from gravitation and all matter fields?

 In special relativity, the energy integral of $T$ on $\Omega$ depends only on the boundary data by energy conservation. One expects energy conservation in general relativity as well, and thus this information should be encoded in the geometry the two-dimensional boundary surface $\Sigma=\partial\Omega$.

 This leads to the well-known problem of quasilocal energy/mass in general relativity. The first one in Penrose's 1982 list \cite{pe} of major unsolved problems in classical
general relativity is ``Find a suitable quasilocal definition of energy-momentum in general relativity".

\section{Total energy and mass}

 Einstein's field equation is derived from variation of the Einstein-Hilbert action on a spacetime domain $M$:
    \[\frac{1}{16\pi}\int_M R+\frac{1}{8\pi} \int_{\partial M} K+\int_M L(g, \Phi)\] where $K$ is the trace
    of the second fundamental form of $\partial M$ and $\Phi$ represents all the matter fields.

 Formally applying Hamilton-Jacobi analysis to this action, we obtain ${T^*}_{\mu\nu}$, the so called Einstein pseudo tensor,
which is expressed in terms of first derivatives of $g$ and satisfies $\nabla^\mu {T^*}_{\mu\nu}=0$.

Here is Hermann Weyl's (1921) comment on $T^*_{\mu\nu}$ (the English translation is quoted from \cite{ch}):

``Nevertheless it seems to be physically meaningless to introduce the $T^*_{\mu\nu}$ as energy
components of the gravitational field; for, these quantities are neither a tensor
nor are they symmetric. In fact by choosing an appropriate coordinate system all
the $T^*_{\mu\nu}$ can be made to vanish at any given point; for this purpose one only needs
to choose a geodesic (normal) coordinate system. And on the other hand one gets
$T^*_{\mu\nu}\not= 0$ in a 'Euclidean' completely gravitationless world when using a curved coordinate
system, but where no gravitational energy exists. Although the differential relations
($ \nabla^\mu {T^*}_{\mu\nu}= 0$) are without a physical meaning, nevertheless by integrating
them over an {\it isolated system} one gets invariant conserved quantities".

An isolated system is modeled on an unbounded and asymptotically flat spacetime where gravitation is weak at infinity.
There are two notions of total mass associated with such a system, one at spatial infinity and the other at null infinity.

Arnowitt-Deser-Misner \cite{adm} applied Hamilton-Jacobi analysis of the Einstein-Hilbert action to such a system that is asymptotically flat at spatial infinity, and obtained a total energy-momentum that is
conserved.

Suppose $({\Omega},g_{ij}, p_{ij})$ is asymptotically flat, i.e. there is a compact subset $K$ of $\Omega$ such that ${\Omega} \backslash K$ is a finite union of ball complements in $\mathbb{R}^3$, and on each component there is asymptotically flat coordinate system such that $g_{ij}-\delta_{ij}\sim 0 $ and $p_{ij}\sim 0$ with appropriate decay rate on their derivatives.
 The total energy is \[E=\lim_{r\rightarrow \infty}\frac{1}{16\pi }
\int_{S_r}(\partial_j g_{ij}-\partial_i g_{jj})dv^i,\] where $S_r$ is the coordinate sphere of coordinate radius $r$.
 The total momentum is \[P_k=\lim_{r\rightarrow \infty} \frac{1}{16\pi }\int_{S_r}
2(p_{ik}-\delta_{ik} p_{jj})dv^i.\]
 $(E, P_1, P_2, P_3)$ is the so called ADM energy momentum 4-vector.

The positive mass theorem of Schoen and Yau \cite{sy1, sy2}(see also Witten \cite{wi}) states that the total mass of such an isolated system is always positive. Suppose the dominant energy condition holds along an asymptotically flat ${\Omega}$, then $(E, P_1, P_2, P_3)$ is a future-directed non-spacelike vector, i.e.
  \[E\geq 0,-E^2+P_1^2+P_2^2+P_3^2\leq0.\]
 In particular, the ADM mass $\sqrt{E^2-P_1^2-P_2^2-P_3^2}$ is non-negative and $=0$ if and only if the spacetime is flat along $\bar{\Omega}$.

 There is also the Bondi-Sachs energy-momentum \cite{bvm} for a asymptotically null hypersurface which measures energy after radiation. Positive energy theorem at null infinity also holds \cite{sy4, hp}, and thus the physical system cannot radiate away more energy than it has initially.

\section{Quasilocal energy/mass and expectations}

We formulate the question of quasilocal energy and mass:

\begin{question} Suppose $\Omega$ is a bounded spacelike region, what is the total energy intercepted by $\Omega$ as seen by an observer? What is the total mass contained in $\Omega$? The answer to these questions should depend only on $\Sigma=\partial\Omega$ by conservation law.

\end{question}
 In comparison to the ADM or Bondi total mass for an isolated system where gravitation is weak at boundary (infinity), the notion of quasilocal mass corresponds to a non-isolated system where gravitation could be strong. What properties qualify for a valid definition? Here are three that we think are most natural:
\begin{arabiclist}
  \item Asymptotics: The limit should recover the ADM mass in the asymptotically flat case and the Bondi mass in the asymptotically null case. It should also recover the energy-momentum tensor in non-vacuum and the Bel-Robinson tensor in vacuum for small sphere limits.

 \item Positivity: The mass should be positive under local energy condition for a large class of
surfaces.

 \item Rigidity: The quasilocal mass should vanish for surfaces in $\mathbb{R}^{3,1}$.
\end{arabiclist}
\section{Hamilton-Jacobi approach}
 There have been various approaches in attempt to define quasilocal mass (see \cite{sz} and the reference therein). We focus on the canonical
Hamilton-Jacobi analysis approach which seems most relevant to Einstein's equation.
Quasilocal Hamilton-Jacobi analysis of Einstein-Hilbert action has been studied by Brown-York \cite{by1, by2}, Hawking-Horowitz \cite{hh}, and Kijowski\cite{ki}.
 Applying the analysis to the time history of a spatially bounded region in spacetime yields the Hamiltonian
 which is a 2-surface integral at terminal time that depends on a pair of vector fields $(t^\mu, u^\mu)$ along $\Sigma$.
$t^\mu$ is a future timelike unit vector field and $u^\mu$ a future timelike unit normal vector. $u^\mu$ should be considered as the future unit normal of a spacelike hypersurface $\Omega$ bounded by $\Sigma$.
We decompose \[t^\mu=N u^\mu+N^\mu.\] The  surface Hamiltonian in \cite{hh} is \begin{equation}\label{surface_ham} \mathcal{H}(t^\mu, u^\mu)=-\frac{1}{8\pi}\int_\Sigma N k-N^\mu v^\nu (p_{\mu\nu}-p_\lambda^\lambda g_{\mu\nu})\end{equation} where
 $k$ is the mean curvature of $\Sigma$ as boundary of $\Omega$,
 $p_{\mu\nu}$ is the second fundamental form of $\Omega$ in spacetime, and
 $v^\nu$ is the outward unit spacelike normal along $\Sigma$  that is orthogonal to $u^\nu$.

 The energy is defined to be the difference between the physical surface Hamiltonian and the reference surface Hamiltonian. Reference surface Hamiltonian in principle should come from data associated with isometric embedding of the time history of the boundary into a reference spacetime. But this is in general an over-determined problem.

 Isometric embedding of $\Sigma$ into $\mathbb{R}^3$ has been used to define Brown-York mass and Liu-Yau mass (see also Kijowski \cite{ki}, Booth-Mann \cite{bm}, Epp \cite{ep}, etc.) with $u^\mu=t^\mu$ (thus $N=1$ and $N^\mu=0$) to be specified. There is a unique isometric embedding into $\mathbb{R}^3$ for any metric with positive Gauss curvature, see Nirenberg \cite{ni} and Pogorelov \cite{po}.

 The Brown-York mass is defined to be $\frac{1}{8\pi}(\int_\Sigma k_0-\int_\Sigma k)$ where $k$ is the mean curvature of $\Sigma$ with respect to a spacelike region $\Omega$, and $k_0$ is the mean curvature of the image of the isometric embedding of $\Sigma$ into $\mathbb{R}^3$. The Liu-Yau mass is $\frac{1}{8\pi}(\int_\Sigma k_0-\int_\Sigma |H|)$ where $H$ is the mean curvature vector of $\Sigma$ in spacetime. Note that the Liu-Yau mass is gauge independent.

 The Brown-York mass and the Liu-Yau mass have the important positivity property by the work of Shi-Tam \cite{st} and Liu-Yau \cite{ly, ly2}, respectively. However, there exist surfaces in $\mathbb{R}^{3,1}$ with strictly positive Brown-York mass and Liu-Yau mass \cite{ost}.

\section{New definition of quasilocal energy}

 For an isometric embedding $X:\Sigma\rightarrow \mathbb{R}^{3,1}$ and $T_0 \in\mathbb{R}^{3,1}$ a constant future timelike unit vector, we define the quasilocal energy to be

\[{ E(\Sigma, X, T_0)=\mathcal{H}(t^\mu, u^\mu)-\mathcal{H}(t_0^\mu, u_0^\mu)}\] where $t_0^\mu=T_0$.

We shall call $\Sigma\subset M$ the physical surface and the image of $X$ in $\mathbb{R}^{3,1}$ the reference surface.

 In the following, we discuss our prescription for  $u_0^\mu$, $t^\mu$, and $u^\mu$ in \cite{wy1, wy2}. Consider the reference surface $\Sigma\subset \mathbb{R}^{3,1}$ and $t_0^\mu$ a constant future timelike unit vector. We take { $u_0^\mu$} to be the unit { normal} future timelike unit vector field in the direction of the normal part of $t_0^\mu$, i.e. $t_0^\mu=N u_0^\nu+N^\mu$ where $N^\mu$ is tangent to $\Sigma$.

 This defines the reference Hamiltonian $\mathcal{H}(t_0^\mu, u_0^\mu)$ which is shown to be equal to
\[-\frac{1}{8\pi}\int_{\hat{\Sigma}}\hat{k}\] where $\hat{\Sigma}$ is the projection of $\Sigma$ onto the orthogonal complement of $t_0^\mu=T_0$.

 We proved a unique isometric embedding theorem \cite{wy2} into $\mathbb{R}^{3,1}$ with convex shadows, i.e. $\hat{\Sigma}$ is a convex surface in the orthogonal complement $\mathbb{R}^3$. To find the corresponding gauge $(t^\mu, u^\mu)$ on the physical surface,
 we assume the { mean curvature vector} of $\Sigma$ in spacetime is spacelike.  For a reference isometric embedding $X:\Sigma\rightarrow \mathbb{R}^{3,1}$ and a $t_0^\mu$, we claim there exists a unique future timelike unit vector { $t^\mu$} along the physical surface $\Sigma\subset M$ such that

 {\it ``The expansion of $\Sigma$ along $t_0^\mu$ in $\mathbb{R}^{3,1}$ is the same as the expansion of $\Sigma$ along $t^\mu$ in $M$".}

 Now define { $u^\mu$} by $t^\mu=Nu^\mu+N^\mu$ along the physical surface $\Sigma\subset M$ for the same $N$ and $N^\mu$. Thus $t^\mu$ and $t^\mu_0$ have the same lapse functions and shift vectors along the physical surface $\Sigma\subset M$ and the reference surface $\Sigma\subset \mathbb{R}^{3,1}$, respectively. Use this $(t^\mu, u^\mu)$ on $\Sigma\subset M$ to compute the physical Hamiltonian $\mathcal{H}(t^\mu, u^\mu)$ and this defines our
{ quasilocal energy $E(\Sigma, X, T_0)$}.

\section{The expression and properties}

 Let $\Sigma$ be a spacelike 2-surface in spacetime which bounds a spacelike hypersurface $\Omega$ with a future unit timelike normal vector field $u^\mu$. Denote by $v^\mu$ the unit spacelike outward normal of $\Sigma=\partial \Omega$ with respect to $\Omega$. The { mean curvature vector} of $\Sigma$ is
\[H=-k v^\mu+p u^\mu\] where $k$ is the mean curvature of $\Sigma$ in $\Omega$ with respect to $v^\mu$ and $p$ is the trace of the restriction of $p_{ij}$ to $\Sigma$. The definition of $H$ is indeed independent of $\Omega$ and the choice of $u^\mu$ and $v^\mu$. Let $J$ be the reflection of $H$ along the future inward null direction in the normal bundle, i.e. $J=k u^\mu-pv^\mu$. $H$ is inward spacelike if and only if $J$ is future timelike.

It turns out $E(\Sigma, X, T_0)$ can be expressed in term of the mean curvature vector field $H$ of $\Sigma$ in $M$ and $\tau=- \langle X, T_0\rangle_{\mathbb{R}^{3,1}}$.

Suppose $H$ is spacelike, we can use the frame $H$ and $J$ to define a connection one-form for the normal bundle of $\Sigma$ by $\langle\nabla_{(\cdot)}^M \frac{J}{|H|}, \frac{H}{|H|}\rangle$.
We recall the following fact that {\it ``the mean curvature vector of the isometric embedding $X:\Sigma\rightarrow \mathbb{R}^{3,1}$ is $H_0=\Delta X$".} Here $\Delta$ is the Laplace operator for functions on $\Sigma$ with respect to the induced metric. For a function defined on $\Sigma$ such as $\tau$, we also use $\nabla\tau$ to denote its gradient vector that is tangent to $\Sigma$.
The quasilocal energy $E(\Sigma, X , T_0)$ with respect to $(X, T_0)$ is
\[\begin{split}&\frac{1}{8\pi}\int_{\hat{\Sigma}}\hat{k}-\frac{1}{8\pi}\int_\Sigma [\sqrt{|H|^2(1+|\nabla \tau|^2)+(\Delta\tau)^2}\\
&-\Delta \tau \sinh^{-1}\frac{\Delta\tau}{\sqrt{1+|\nabla \tau|^2}|H|}
-\langle\nabla^{M}_{\nabla\tau} \frac{J}{|H|}, \frac{H}{|H|}\rangle ] \end{split}\]where
 \[\begin{split}&\int_{\hat{\Sigma}}\hat{k}=\int_\Sigma [\sqrt{|H_0|^2(1+|\nabla \tau|^2)+(\Delta\tau)^2}\\
&-\Delta \tau \sinh^{-1}\frac{\Delta\tau}{\sqrt{1+|\nabla \tau|^2}|H_0|}
-\langle\nabla^{\mathbb{R}^{3,1}}_{\nabla\tau} \frac{J_0}{|H_0|}, \frac{H_0}{|H_0|}\rangle ].\end{split}\]

 Quasilocal mass is defined to be the { infimum of quasilocal energy} $E(\Sigma, X, T_0)$ among all ``admissible observers" $(X, T_0)$ (see \cite{wy1} for the definition):
\[m(\Sigma)=\inf E(\Sigma, X, T_0).\]

In \cite{wy1, wy2, wy3, cwy}, we prove:

\begin{arabiclist}
\item { Positivity}: $m(\Sigma)\geq 0 $ under dominant energy condition on spacetime and convexity assumptions on $\Sigma$.

\item { Rigidity}: $m(\Sigma)=0$ if $\Sigma$ is in $\mathbb{R}^{3,1}$.

\item Quasilocal mass approaches the { ADM mass} and { Bondi mass} at spatial and null infinity, respectively.
\end{arabiclist}
This is the only known definition of quasilocal mass that satisfies all these properties.

In fact, the quasilocal energy $E(S_r, X_r, T_0)$ gets { linearized} and acquires the Lorentzian symmetry at infinity.
 \[\lim_{r\rightarrow \infty} E(S_r, X_r, T_0)=T_0^\mu P_{\mu}\] where $P_\mu=(P_0, P_1, P_2, P_3)$ is the ADM / Bondi-Sachs
{ energy-momentum 4-vector}, at spatial/null infinity.

 In general, suppose $\Sigma_r$ is a family of surface in spacetime and a family of isometric embedding $X_r$ of $\Sigma_r$ in $\mathbb{R}^{3,1}$ is given. As long as $\frac{|H_0|}{|H|}\rightarrow 1$ as $r\rightarrow \infty$, the limit of the quasilocal energy $E(\Sigma_r, X_r, T_0)$ is the same as the limit of
\[\frac{1}{8\pi}  \int_{\Sigma_{r}} -\langle T_0,  \frac{{J}_0}{|H_0|}\rangle (|H_0|-|H|) - \langle\nabla^{\mathbb{R}^{3,1}}_{\nabla\tau} \frac{{J}_0}{|H_0|}, \frac{H_0}{|H_0|}\rangle+ \langle\nabla^{N}_{\nabla\tau} \frac{{J}}{|H|}, \frac{H}{|H|}\rangle.\]
As $\tau=-\langle X, T_0\rangle$, the expression is already { linear} in $T_0$.

\section{Surface Hamiltonian and Minkowski inequality}

In this section, we discuss the surface Hamiltonian in Minkowski space and the connection to an inequality proposed by Gibbons.

We recall the following identity in \cite{wy2, wa1} regarding the surface Hamiltonian \eqref{surface_ham}:
\begin{proposition}\label{ham_min}
 For a closed spacelike 2-surface $\Sigma$ in the Minkowski space which bounds a spacelike hypersurface and a constant future timelike unit vector field $T_0$, there exists a unique orthogonal normal gauge $\{\breve{e}_3, \breve{e}_4\}$ along $\Sigma$ such that $\breve{e}_3$ is a outward spacelike unit normal and $\breve{e}_4$ is a future timelike unit normal and they satisfy
\begin{equation}\label{ref_gauge}-\frac{1}{8\pi}\int_\Sigma \langle J, T_0\rangle_{\mathbb{R}^{3,1}}+\langle \nabla^{\mathbb{R}^{3,1}}_{\breve{e}_3} \breve{e}_4, T_0^\top \rangle_{\mathbb{R}^{3,1}} =\frac{1}{8\pi}\int_{\hat{\Sigma}} \hat{k}\end{equation} where $\hat{\Sigma}$ is the projection of $\Sigma$ onto the orthogonal complement of $T_0$ and $\hat{k}$ is the mean curvature of $\hat{\Sigma}$.
\end{proposition}

\begin{proof}
Proposition 3.1 of \cite{wy2} (see also \cite{gi}).
\end{proof}

 In fact, denote by $\tau$ the restriction of time function defined by $T_0$ to $\Sigma$ and by $\nabla\tau$ the gradient vector field of $\tau$ on $\Sigma$ with respect to the induced metric, we have \[T_0=\sqrt{1+|\nabla\tau|^2}\breve{e}_4-\nabla \tau.\] The lapse and shift of $T_0$ are given by $ \sqrt{1+|\nabla\tau|^2}$  and  $ T_0^\top=-\nabla\tau$, respectively.

The classical Minkowski inequality for surfaces in $\mathbb{R}^3$ states that for a closed convex surface $\hat{\Sigma}$ in $\mathbb{R}^3$,
\[\int_{\hat{\Sigma}} \hat{k} \, d\mu \geq \sqrt{16 \pi \, |\hat{\Sigma}|},\]
where $\hat{k}$ is the mean curvature and $|\hat{\Sigma}|$ is the area of $\hat{\Sigma}$.

Applying the Minkowski inequality and recalling that the area of $\hat{\Sigma}$ is always greater than or equal to the area of $\Sigma$, we obtained the following inequality between the surface Hamiltonian and the area.

\begin{equation}\label{Hamiltonian_area}-\frac{1}{8\pi}\int_\Sigma \langle J, T_0\rangle_{\mathbb{R}^{3,1}}+\langle \nabla^{\mathbb{R}^{3,1}}_{\breve{e}_3} \breve{e}_4, T_0^\top \rangle_{\mathbb{R}^{3,1}} \geq \sqrt{\frac{|\Sigma|}{4\pi}}.\end{equation}

In equation (6.16) of \cite{gi}, the author claimed that the following inequality holds and called it the black hole isoperimetric inequality.

\begin{equation}\label{bh}-\frac{1}{8\pi}\int_\Sigma \langle J, T_0\rangle_{\mathbb{R}^{3,1}} \geq \sqrt{\frac{|\Sigma|}{4\pi}}.\end{equation}

However, the derivation in \cite{gi} is not correct and the validity of this inequality remains open, see also section 7.1 of \cite{ma}.

Recently, a sharp Minkowski inequality in the hyperbolic 3-space was proved in \cite{bhw}:
\begin{equation}\label{min_hyper}\int_\Sigma f \, h \, d\mu - 6 \int_\Omega f \, d\text{\rm vol}\geq \sqrt{16 \pi \, |\Sigma|}\end{equation} for any mean convex, star shaped region $\Omega \subset \mathbb{H}^3$ and $\Sigma=\partial\Omega$. Here $f=\cosh r$ where $r$ is the geodesic distance function with respect to a point $o\in \mathbb{H}^3$ and $h$ is the mean curvature of $\Sigma$ with respect to outward unit normal of $\Omega$.

In the following, we show that \eqref{min_hyper} is equivalent to \eqref{bh} when the surface $\Sigma\subset \mathbb{R}^{3,1}$ lies in the hyperbolic space $\mathbb{H}^3\subset \mathbb{R}^{3,1}$.

\begin{theorem} Suppose $\Sigma$ is a closed embedded spacelike 2-surface in the Minkowski space and $T_0\in \mathbb{H}^3\subset \mathbb{R}^{3,1}$ is a future unit timelike constant vector. Suppose $\Sigma$ lies in $\mathbb{H}^3$ and is mean convex and star-shaped with respect to $T_0$, then inequality \eqref{bh} holds, i.e.
\[-\frac{1}{8\pi}\int_\Sigma \langle J, T_0\rangle_{\mathbb{R}^{3,1}} \geq \sqrt{\frac{|\Sigma|}{4\pi}}.\]
\end{theorem}
\begin{proof} By Lorentz transformation, we may assume $T_0$ is $(1,0,0,0)$. Now choose $r$ to be the geodesic distance function on $\mathbb{H}^3$ with $(1,0,0,0)$ as the origin.
 We embed $\mathbb{H}^3$ into $\mathbb{R}^{3,1}$ as the upper branch of the hyperbola $\{ (t, x, y, z)\,|\, t>0, -t^2+x^2+y^2+z^2=-1\}$. Suppose the embedding is given by the position four-vector $X=(t, x, y, z)$

 The mean curvature vector of $\Sigma$ in $\mathbb{R}^{3,1}$ is $-h\nu +2e_4$ where $h$ is the mean curvature of $\Sigma$ in $\mathbb{H}^3$, $\nu$ is the outward unit normal of $\Sigma$ in $\mathbb{H}^3$, and $e_4$ is the future unit timelike normal of $\mathbb{H}^3$ in $\mathbb{R}^{3,1}$. Note that $e_4$ is the same as the position vector of the embedding $X:\mathbb{H}^3\rightarrow \mathbb{R}^{3,1}$.

 Consider the normal vector field \[{J}=he_4-2\nu\] obtained by reflecting ${H}$ along the future inward null direction of the normal bundle.  We check that
\[-\int_\Sigma \langle {J},  T_0 \rangle_{\mathbb{R}^{3,1}}\,\,d\mu=\int_\Sigma ( f \,\,h-2 \frac{\partial f}{\partial \nu}) \,\, d\mu=\int_\Sigma f \, h \, d\mu - 6 \int_\Omega f \, d\text{\rm vol}\] because $f=-\langle X, T_0\rangle$ and $\bar{\Delta} f=3 f$ where $\bar{\Delta}$ is the Laplace operator on $\mathbb{H}^3$.
\end{proof}

If we assume the surface has spacelike inward mean curvature vector, the integral $-\int_\Sigma \langle J, T_0\rangle_{\mathbb{R}^{3,1}}\,\, d\mu$ is positive for any $T_0$.
 From this, we can formulate a general question for spacelike surfaces in $\mathbb{R}^{3,1}$.
 \begin{question} Suppose $\Sigma$ is closed spacelike 2-surface that bounds a spacelike hypersurface in $\mathbb{R}^{3,1}$. Under what condition does  inequality \eqref{bh} hold? \end{question}
\section*{Acknowledgements}

 The author would like to thank Simon Brendle, PoNing Chen, Gary Gibbons, and Shing-Tung Yau for helpful discussions.
The author's research was supported by the National Science Foundation under grant DMS-1105483.

\end{document}